\documentclass[10pt,conference]{IEEEtran}
\normalsize
\usepackage{amssymb,amsmath,amsthm,bm}
\usepackage{enumitem}
\usepackage{flushend}

\usepackage[ruled]{algorithm}
\usepackage[noend]{algorithmic}
\makeatletter
\newcommand\fs@norules{\def\@fs@cfont{\bfseries}\let\@fs@capt\floatc@ruled
	\def\@fs@pre{}%
	\def\@fs@post{}%
	\def\@fs@mid{\kern3pt}%
	\let\@fs@iftopcapt\iftrue}
\makeatother

\usepackage[utf8]{inputenc}
\usepackage[T1]{fontenc}

\usepackage{graphicx}
\usepackage{xcolor} 
\usepackage{tikz,tikzscale}
\usetikzlibrary{fit, spy, calc, shapes, arrows, positioning, plotmarks, arrows.meta, pgfplots.groupplots, backgrounds}
\usepackage{pgfplots}
\pgfplotsset{compat=newest}
\usepgfplotslibrary{patchplots}

\makeatletter
\let\MYcaption\@makecaption
\makeatother
\usepackage[font=footnotesize]{subcaption}
\makeatletter
\let\@makecaption\MYcaption
\makeatother

\renewcommand{\vec}[1]{{\bm{#1}}}
\newtheorem{definition}{Definition}
\newtheorem{example}{Example}
\newtheorem{lemma}{Lemma}
\DeclareMathOperator*{\argmin}{argmin}
\usepackage{mathtools}
\newcommand{\defeq}{\vcentcolon=}

\begin{document}
\bstctlcite{IEEEexample:BSTcontrol}
\title{Arithmetic Coding Based Multi-Composition Codes for Bit-Level Distribution Matching } 
\author{ 
	Marcin~Pikus\IEEEauthorrefmark{1}\IEEEauthorrefmark{2}
	and~Wen~Xu\IEEEauthorrefmark{1}\\
	\IEEEauthorblockA{\IEEEauthorrefmark{1}Huawei Technologies, Munich Research Center, Riesstr. 25, 80992 Munich, Germany\\
		\IEEEauthorblockA{\IEEEauthorrefmark{2}Institute for Communications Engineering, Technische Universit{\"{a}}t M{\"{u}}nchen, Arcisstr. 21, 80290 Munich, Germany}\vspace{-40pt}
	}

\thanks{M. Pikus is with the Huawei Technologies Duesseldorf GmbH, D-80992 Munich, Germany,
	 and also with the Institute for Communications Engineering, Technische Universität München, D-80333
	Munich, Germany (e-mail: marcin.pikus@gmail.com).}
\thanks{W. Xu is with the Huawei Technologies Duesseldorf GmbH, D-80992 Munich, Germany (e-mail: wen.xu@ieee.org).}
}
\maketitle

\begin{abstract}
A distribution matcher (DM) encodes a binary input data sequence into a sequence of symbols (codeword) with desired target probability distribution. The set of the output codewords constitutes a codebook (or code) of a DM. Constant-composition DM (CCDM) uses arithmetic coding to efficiently encode data into codewords from a constant-composition (CC) codebook. The CC constraint limits the size of the codebook, and hence the coding rate of the CCDM. The performance of CCDM degrades with decreasing output length. To improve the performance for short transmission blocks we present a class of multi-composition (MC) codes and an efficient arithmetic coding scheme for encoding and decoding. The resulting multi-composition DM (MCDM) is able to encode more data into distribution matched codewords than the CCDM and achieves lower KL divergence, especially for short block messages.
\end{abstract}


\IEEEpeerreviewmaketitle

\section{Introduction}
A distribution matcher (DM) reversibly maps a sequence $\vec{U}$ of independent and uniformly distributed bits into a sequence $\vec{A}$ of symbols to emulate a target distribution $P_A$.  The output of the DM approximates a sequence of independent and identically distributed (IID) symbols, each distributed according to $P_A$.  The accuracy of the approximation is measured by the Kullback–Leibler (KL) divergence between the probability distribution of the DM's output and the probability distribution of the IID sequence.  An inverse distribution matcher (DM$^{-1}$) performs the inverse operation recovering $\vec{U}$  from $\vec{A}$. DMs can be used in communication systems, such as probabilistic amplitude shaping (PAS) \cite{7307154}, to adjust the distribution of transmitted symbols to a distribution beneficial for a certain channel, e.g., a distribution achieving capacity, or reducing the peak-to-average-power ratio. PAS was recently proposed for the 5G mobile system \cite{PSCM_Huawei}.  

We focus here on block-to-block (b2b) DMs, where the input and output sequences have fixed lengths denoted by $k$, and $n$, respectively. The ratio $k/n$ is called  the \textit{matching rate}. Variable length DMs, i.e., non b2b DMs, lead to synchronization problems, error propagation, and variable transmission rate \cite{b2bdm}. A DM can be seen as an encoder which maps $k$ input bits to a non-binary codeword of length $n$. The set of possible codewords, i.e., the codebook (or code), is chosen such that a certain output distribution is emulated. To achieve a good performance, i.e., high matching rate and low normalized KL divergence, relatively long output sequences are needed \cite{7322261}. Long sequences imply a large codebook which may not be stored in memory. DMs which need to store the codebooks, e.g., look-up-table mappings, have thus limited performance. 

A Constant-Composition DM (CCDM) \cite{7322261} is a DM which employs arithmetic coding to generate codewords on-the-fly. In this way, a large codebook can be used without the need of storage. The CCDM is asymptotically optimal for $n \to \infty$, i.e., it achieves the maximal matching rate and vanishing normalized KL divergence. For finite $n$ however, the CCDM suffers from rate loss and high divergence that increases with decreasing $n$. The CCDM uses a constant-composition (CC) codebook, which constraints each codeword to have a fixed number of occurrences of each of the symbols from the alphabet, i.e., each codeword has the same \textit{composition}. The CC constraint contributes to increased KL divergence and decreased matching rate by limiting the number of codewords. It is therefore of interest to improve the DM performance for smaller values of $n$ for at least two reasons: 1) low performance DM, e.g., with high rate loss, can waive the benefits of PAS;  2) throughput and parallelization---it allows to replace one DM with output length $n$ by $\alpha$ parallel DMs with output lengths $n/\alpha$ without significant decrease in the performance. This leads to increased throughput of PAS systems, as DM is the throughput limiting component in the PAS transmission chain. 

In this work, we introduce a Multi-Composition (MC) codes which can be used to build a MC DM (MCDM). A MC codebook is a codebook which contains codewords of multiple compositions. By a proper construction of the MC codebook, the MCDM is able to employ an arithmetic coding algorithm to generate codewords from the MC codebook on-the-fly. The MCDM generalizes the CCDM, and the aforementioned asymptotic optimality of the CCDM holds for the MCDM. By relaxing the CC constraint, we obtain a DM which achieves higher matching rate and lower KL divergence for any $n$. Just recently, two other solutions have also been proposed in the literature, i.e., partition based DM \cite{partitiondm} and shell mapping DM (SMDM) \cite{shelldm}. Roughly speaking, our MCDM is able to index more codewords than \cite{partitiondm}, thus leading to better performance. The proposed MCDM can also be implemented with low-complexity arithmetic coding requiring $\mathcal{O}(n)$ computational complexity \cite{58748}, and can operate with $k$ and $n$ of arbitrary length. For the binary case, we present a MCDM whose performance can approach that of the optimal b2b DM. 

The SMDM is optimal DM, which is suitable for short output lengths due to its complexity \cite{shelldm}. The CCDM is asymptotically optimal for long output lengths \cite{7322261}. The proposed MCDM is based on arithmetic coding just as CCDM. This way, one algorithm (with different parameters) can be used for short and long sequences. This reduces the complexity of the system. 

The structure of this work is organized as follows. In Sec. \ref{s:dm_fund} we introduce distribution matching and the CCDM. In Sec. \ref{s:acoding} we briefly describe how arithmetic coding can be applied for distribution matching. In Sec. \ref{s:mcc} we present the MCDM which can be efficiently implemented with the arithmetic coding scheme. Simulation results are given in Sec. \ref{s:results}. Finally, conclusions are drawn in Sec. \ref{s:conclusions}.

We use the following notations. We denote random variables (RVs) by capital letters, such as $A$, and realizations by small letters, such as $a$. A row vector is denoted by a bold symbol, e.g., $\vec{a}$. The $i$-th entry in the vector $\vec{a}$ is denoted by $a_i$, and a subvector $[a_i,a_{i+1},\cdots, a_j]$ of $\vec{a}$ is denoted by $\vec{a}_i^j$. The length of a vector is denoted by $l(\vec{a})$. E.g., we have $\vec{a} = \vec{a}_1^{l(\vec{a})}$. A RV uniformly distributed on a set $S$ is denoted by $\mathbb{U}_S$, i.e., $S \sim \mathbb{U}_S$ means that $P_{\mathbb{U}_S}(s)=1/|S|$ for $s \in S$.

\section{Distribution Matching} \label{s:dm_fund}
A one-to-one b2b DM is an injective function $f_{\text{DM}}$ from binary input sequences $\vec{u} \in \{0,1\}^k$ to codewords $\vec{c}$ from the codebook $\mathcal{C} \subseteq  \mathcal{A}^n$, i.e., 
\begin{equation}
f_{\text{DM}} \colon  \{0,1\}^k \to \mathcal{C}, 
\end{equation}
where $\mathcal{A}$ is the output alphabet. We assume that the input sequence $\vec{U}$ is a random vector consisting of $k$ IID Bernoulli$(1/2)$ distributed bits. The output sequence of the DM is thus a random vector $\vec{\tilde{A}} = f_{\text{DM}}(\vec{U}) \sim \mathbb{U}_{\mathcal{C}}$ uniformly distributed on $ \mathcal{C}$. The goal of the DM is to make its output "look" as if it was a sequence of IID RVs, each distributed according to the target probability distribution $P_A$. This is usually performed by minimizing the normalized KL divergence between the DM's output $\vec{\tilde{A}}$ and the IID sequence $\vec{A} \sim P_A^n=\prod_{i=1}^n P_A$
\begin{equation}
\frac{1}{n}\mathbb{D}(P_{\vec{\tilde{A}}} \| P_A^n) = \frac{1}{n} \sum_{\vec{c} \in \mathcal{C}}  \frac{1}{|\mathcal{C}|} \log \frac{\frac{1}{|\mathcal{C}|}}{ P_A^n(\vec{c})}. \label{eq:def_div}
\end{equation} 
The divergence can thus be minimized by choosing a proper codebook $\mathcal{C}$. The empirical probability of a single symbol outputted by a DM is defined as
\begin{equation}
P_{\mathcal{C}}(a) = \frac{\sum_{\vec{c} \in \mathcal{C}} n_a(\vec{c})}{n|\mathcal{C}|},
\end{equation}
where we use $n_a(\vec{c}) \defeq \left|\{i \colon  c_i = a \}\right|$ to denote the number of occurrences of $a$ in the sequence $\vec{c}$ and $|\mathcal{C}|$ denotes the number of codewords in the codebook $\mathcal{C}$ (size of the codebook). In the literature, it is often believed that we need $P_{\mathcal{C}} = P_A$ (or $P_{\mathcal{C}} \approx P_A$) to minimize the divergence for finite $n$. However, this may not necessarily be true as we shall see in Sec. \ref{s:results} or as pointed out by authors in \cite[Example 2]{shelldm}. 
The divergence (\ref{eq:def_div}) can be equivalently written as 
\begin{equation}
\frac{1}{n}\mathbb{D}(P_{\vec{\tilde{A}}} \| P_A^n) = \mathbb{H}(P_{\mathcal{C}}) - \frac{\log |\mathcal{C}|}{n} + \mathbb{D}(P_{\mathcal{C}}||P_A), \label{eq:def_div2}
\end{equation} 
where $\mathbb{H}(P_{\mathcal{C}})$ in the entropy of a RV with distribution $P_{\mathcal{C}}$. We observe that the necessary condition for vanishing normalized divergence is that $P_{\mathcal{C}} \to P_A$\footnote{Since $\mathbb{H}(P_{\mathcal{C}}) - \frac{\log |\mathcal{C}|}{n} = \frac{1}{n}\mathbb{D}(P_{\vec{\tilde{A}}} || P_{\mathcal{C}}^n) \ge 0$.}, however $\frac{1}{n}\mathbb{D}(P_{\vec{\tilde{A}}} \| P_A^n) \to 0$ only for $n \to \infty$ \cite{div_scaling}. Thus, we can only expect $P_{\mathcal{C}} \approx P_A$ for large $n$. Equation (\ref{eq:def_div2}) also suggests that for a given $P_{\mathcal{C}}$, larger codebooks, i.e., with greater $|\mathcal{C}|$, are preferred.
 
In \cite{7929328} it was shown that non-binary distribution matching, i.e., with $|\mathcal{A}|>2$, can be well approximated by multiple, binary distribution matchings, i.e., with $|\mathcal{A}|=2$, for typical use cases. For simplicity, we focus on the binary distribution matching in the following sections. The MCDM can be also directly used for non-binary distribution matching, as explained is Sec. \ref{ss:nonbinary}.

CCDM was introduced in \cite{7322261}. It uses a modified coding scheme \cite{58748} based on arithmetic coding to efficiently encode data into the codewords from a CC codebook. In the CC codebook each codeword has the same \textit{composition}.
\begin{definition}
	Assume $\mathcal{A} \!=\! \{a_1, \dotsc, a_m\}$. A \textit{composition} of a vector $\vec{c} \in \mathcal{A}^n$ is a vector containing the numbers of occurrences in $\vec{c}$ of each of the symbols from the alphabet $\mathcal{A}$. We denote a composition by 
	\begin{equation}
	\vec{\gamma}(\vec{c}) \defeq [n_{a_1}(\vec{c}), \dotsc,  n_{a_m}(\vec{c})].
	\end{equation}
\end{definition}
\begin{example} 
	$\mathcal{A} \!=\! \{0, 1\}, \vec{\gamma}(1011) = [1,3]$.
\end{example}
That is, the CC codebook with the composition $\vec{\gamma}$ is
\begin{equation}
	\mathcal{C}_{\text{CC}} = \{ \vec{c} \in \mathcal{A}^n \colon  \vec{\gamma}(\vec{c}) = \vec{\gamma}\ \},
\end{equation}
and the size of the codebook can be expressed by the multinomial coefficient
\begin{equation}
	|\mathcal{C}_{\text{CC}}| = \binom{n}{\gamma_1, \dots, \gamma_m} = \frac{n!}{\prod_{i=1}^{m}\gamma_i!}.
\end{equation}
To guarantee a one-to-one mapping between the binary input sequences and the codewords, the CCDM can use at most $2^k$  input sequences of length $k$, where $k$ reads as
\begin{equation}
	k = \lfloor \log_2 	|\mathcal{C}_{\text{CC}}| \rfloor,
\end{equation}
where $\lfloor \cdot  \rfloor$ is the floor function.

\section{Arithmetic Coding in Distribution Matching} \label{s:acoding}
The coding scheme \cite{58748} based on arithmetic coding has been proposed to efficiently realize encoding and decoding for the so-called $m$-out-of-$n$ codebooks, i.e., the binary constant-weight codebooks, which are a spacial case of the CC codebooks for binary alphabets. In \cite{7322261} it is shown that the arithmetic coding scheme from \cite{58748} can be utilized for the CC codebooks with non-binary alphabets, i.e., for CCDM implementation. In what follows, we demonstrate that the arithmetic coding scheme presented in \cite{58748} can be also employed to efficiently implement our MCDM. 

For simplicity, we consider the binary output alphabet $\mathcal{A}=\{0,1\}$. Assume an arbitrary codebook $\tilde{\mathcal{C}} \subseteq	\mathcal{A}^n$.  Each input data sequence $\vec{u}_i, i\!=\!1,\dots,2^k$ corresponds to a distinct point $d(\vec{u}_i), i\!=\!1,\dots,2^k$ from the interval $[0,1)$. On the other hand, each codeword $\vec{c} \in \tilde{\mathcal{C}}$ corresponds to a distinct subinterval $I(\vec{c})$ of the interval $[0,1)$. The subintervals $I(\vec{c}), \vec{c} \in \tilde{\mathcal{C}}$ are chosen such that they \textit{partition} the interval $[0,1)$, i.e., they are pairwise disjoint and $\bigcup_{\vec{c} \in \tilde{\mathcal{C}}} I(\vec{c}) = [0,1)$. At the encoder, an input data sequence $\vec{u}$ is mapped to a codeword $\vec{c}$ if the corresponding point $d(\vec{u})$ lines inside the corresponding interval $I(\vec{c})$. At the decoder, first an interval $I(\vec{c})$ is determined based on  the received codeword $\vec{c}$. Then, a point $d(\vec{u}) \in I(\vec{c})$ is determined and decoded to the sequence $\vec{u}$. 

Assume a binary input sequence $\vec{u}=[u_1,\dotsc,u_k]$. Let $\text{NBC}(\cdot)$ denote a function which returns the natural binary code (NBC) number corresponding to the sequence $\vec{u}$, i.e.,
\begin{equation}
\text{NBC}(\vec{u}) = \sum_{j=1}^k u_j2^{k-j}.
\end{equation}
The sequence $\vec{u}$ is mapped to a point $d(\vec{u}) \in [0,1)$ via
\begin{equation}
d(\vec{u}) = \frac{\text{NBC}(\vec{u})}{2^k}. \label{eq:du}
\end{equation}
An interval $I(\vec{c})$ for a codeword $\vec{c}=[c_1, \dotsc, c_n]$ can be computed recursively using a chosen probability model $P_{\vec{C}}$ on codeword's bits. The model $P_{\vec{C}}$ is specified in terms of the conditional probabilities (also called branching probabilities) of the next bit given the previous bits, i.e., $P_{C_i|\vec{C}_1^{i-1}}(\cdot|\vec{s})$, where $\vec{s}$ is a sequence denoting a prefix of the codeword. The beginning $x(\vec{c})$ and the width $y(\vec{c})$ of the interval $I(\vec{c}) = [x(\vec{c}), x(\vec{c})+y(\vec{c}) )$ can be computed by applying iteratively equations (\ref{eq:x_update}) and (\ref{eq:y_update}) for $i=1,\dots,n,$
\begin{align}
&x(\emptyset) =0, y(\emptyset) =1, \label{eq:xy_init}\\
&x(\vec{s}c_i) = \begin{cases} 
x(\vec{s}), \text{ if } c_i=0\\
x(\vec{s}) + y(\vec{s})P(0|\vec{s}), \text{ if } c_i=1
 \end{cases} \label{eq:x_update}\\
&y(\vec{s}c_i) = y(\vec{s}) P(c_i|\vec{s}), \label{eq:y_update}
\end{align}
where $\vec{s}c_i$ denotes a concatenation of $\vec{s}$ and $c_i$,  $\emptyset$ denotes an empty sequence, and $P(0|\vec{s})$ stands for $P_{C_i|\vec{C}_1^{i-1}}(0|\vec{s})$. Equation ($\ref{eq:x_update}$) implies a lexicographical ordering of the codewords according to $0<1$ with the most-significant-bit $c_1$. That is, for two codewords $\vec{c}_1$ and $\vec{c}_2$, if $\text{NBC}(\vec{c}_1)<\text{NBC}(\vec{c}_2)$, the $I(\vec{c}_1)$ will be placed in the interval $[0,1)$ below the $I(\vec{c}_2)$. Applying the above equations result in partitioning such that $\forall \vec{c} \in \tilde{\mathcal{C}}$
\begin{align}
&x(\vec{c}) = \sum_{\vec{c}^\prime \in \tilde{\mathcal{C}} \colon \text{NBC}(\vec{c}^\prime)<\text{NBC}(\vec{c})} P_{\vec{C}}(\vec{c}^\prime), \label{eq:x_value}\\
&y(\vec{c}) = \prod_{i=1}^n P(c_i | \vec{c}_1^{i-1}) = P_{\vec{C}}(\vec{c}).  \label{eq:y_value}
\end{align}
That is, the codewords' intervals partition $[0,1)$ and are ordered according to lexicographical ordering $0<1$. E.g., see Fig. \ref{f:base_cb}.
\addtolength{\topmargin}{0.1in}
\addtolength{\textheight}{-0.1in}
A one-to-one mapping between data sequences and codewords can be established if each interval  $I(\vec{c}), \vec{c} \in \tilde{\mathcal{C}}$ contains at most one point $d(\vec{u})$. This can be guaranteed by letting the distance between two adjacent points to be grater than the largest interval, i.e.,
\begin{equation}
\frac{1}{2^k} \ge \max_{\vec{c} \in \tilde{\mathcal{C}}} |I(\vec{c})|. \label{eq:interval_ineq}
\end{equation} 
Since we are interested in maximizing $k$, it is reasonable to choose equal length intervals, i.e., $|I(\vec{c})| = \frac{1}{|\tilde{\mathcal{C}}|} \forall \vec{c} \in \tilde{\mathcal{C}}$. In this case the greatest $k$ fulfilling (\ref{eq:interval_ineq}) equals $\lfloor{\log_2|\tilde{\mathcal{C}}|}\rfloor$.

From (\ref{eq:y_value}) we have that the length of the interval $I(\vec{c})$ is equal to the probability of the codeword $\vec{c}$ (by using the probability model $P_{\vec{C}}$ on codeword's bits).
We are interested in finding the conditional probabilities $P(0|\vec{s})$ and $P(1|\vec{s})$, where $\vec{s}$ is a binary sequence constituting the prefix of the codeword, such that the probability of each codeword $\vec{c} \in \tilde{\mathcal{C}}$ is equal. Let $N(\vec{s})$ denote the number of codewords in $\tilde{\mathcal{C}}$ that have prefix $\vec{s}$, i.e., $N(\vec{s}) \defeq \left| \{ \vec{c} \in \tilde{\mathcal{C}} \colon \vec{c}_1^{l(\vec{s})} = \vec{s}\}\right|$. We define the following conditional probabilities
\begin{align}
P(0|\vec{s}) \defeq \frac{N(\vec{s}0)}{N(\vec{s})}, \;\; P(1|\vec{s}) \defeq \frac{N(\vec{s}1)}{N(\vec{s})}. \label{eq:pb_s}
\end{align}
For any $\vec{c}=[c_1, \dots, c_n] \in \tilde{\mathcal{C}}$ we have
\begin{align}
P(\vec{c}) &= P(c_1|\emptyset)P(c_2|c_1) \dotsc P(c_n|\vec{c}_1^{n-1})\\
&= \frac{N(c_1)}{N(\emptyset)} \frac{N(\vec{c}_1^2)}{N(c_1)} \dotsc \frac{N(\vec{c}_1^n)}{N(\vec{c}_1^{n-1})} = \frac{N(\vec{c}_1^n)}{N(\emptyset)} = \frac{1}{|\tilde{\mathcal{C}}|},
\end{align}
which shows that by employing model as in (\ref{eq:pb_s}) we can obtain intervals of equal length.

\begin{algorithm}[t]%
	\small
	\caption{\strut Encoding}%
	\begin{algorithmic}[1]%
		\renewcommand{\algorithmicrequire}{\textbf{Input:}}
		\renewcommand{\algorithmicensure}{\textbf{Output:}}
		\REQUIRE binary sequence $\vec{u}=[u_1, \cdots, u_k]$
		\ENSURE  codeword $\vec{c}=[c_1, \cdots, c_n]$ s.t. $d(\vec{u}) \in I(\vec{c})$
		\STATE $\vec{s} = \emptyset$,  $x(\emptyset) = 0$, $y(\emptyset) = 1$ 		
		\FOR {$i = 1$ to $n$}
		\IF {$d(\vec{u}) \in [x(\vec{s}),x(\vec{s}) + y(\vec{s})P(0|\vec{s}))$}
		\STATE $x(\vec{s}0) \gets x(\vec{s})$
		\STATE $y(\vec{s}0) \gets y(\vec{s})P(0|\vec{s})$
		\STATE $\vec{s} \gets \vec{s}0$
		\ELSE
		\STATE $x(\vec{s}1) \gets x(\vec{s}) + y(\vec{s})P(0|\vec{s})$
		\STATE $y(\vec{s}1) \gets y(\vec{s})P(1|\vec{s})$
		\STATE $\vec{s} \gets \vec{s}1$
		\ENDIF
		\ENDFOR
		\RETURN $\vec{c} \gets \vec{s}$
	\end{algorithmic}%
\end{algorithm}%

\begin{algorithm}[t]%
	\small
	\caption{\strut Decoding}%
	\begin{algorithmic}[1]%
		\renewcommand{\algorithmicrequire}{\textbf{Input:}}
		\renewcommand{\algorithmicensure}{\textbf{Output:}}
		\REQUIRE codeword $\vec{c}=[c_1, \cdots, c_n]$ 
		\ENSURE binary sequence $\vec{u}=[u_1, \cdots, u_k]$ s.t. $d(\vec{u}) \in I(\vec{c})$		  
		\STATE $\vec{s} = \emptyset$, $x(\emptyset) = 0$, $y(\emptyset) = 1$ 		
		\FOR {$i = 1$ to $n$}
		\IF {$c_i = 0$}
		\STATE $x(\vec{s}0) \gets x(\vec{s})$
		\STATE $y(\vec{s}0) \gets y(\vec{s})P(0|\vec{s})$
		\STATE $\vec{s} \gets \vec{s}0$
		\ELSE
		\STATE $x(\vec{s}1) \gets x(\vec{s}) + y(\vec{s})P(0|\vec{s})$
		\STATE $y(\vec{s}1) \gets y(\vec{s})P(1|\vec{s})$
		\STATE $\vec{s} \gets \vec{s}1$
		\ENDIF
		\ENDFOR
		\STATE $\vec{u} = \text{NBC}^{-1}(\lceil x(\vec{s})2^k \rceil)$
		\RETURN $\vec{u}$
	\end{algorithmic}%
\end{algorithm}%


To encode data into the codewords from $\tilde{\mathcal{C}}$ we apply arithmetic decompression of the input sequence $\vec{u}$  using the model (\ref{eq:pb_s}), see Algorithm 1 for details.  
To decode the codeword, we apply arithmetic compression on the codeword $\vec{c}$ using the same model (\ref{eq:pb_s}), see Algorithm 2 for details. 
Retrieving $\vec{u}$ from $I(\vec{c})$ is the final step of Algorithm 2, and is performed in the line $11$ which follows from (\ref{eq:du}) and the fact that $d(\vec{u}) \in I(\vec{c})$. In practice, to avoid numerical underflow, the intervals have to be rescaled during each step. The coding scheme can be also implemented using only integer calculations. For further implementation considerations, see e.g. \cite{58748}.

Following the above steps we can efficiently encode/decode data into/from the codewords $\vec{c} \in \tilde{\mathcal{C}}$. In general, not all codewords from $\tilde{\mathcal{C}}$ will be used since, by (\ref{eq:interval_ineq}), we can use at most $2^{\log_2 \lfloor |\tilde{\mathcal{C}}| \rfloor}$ codewords. The selection of the used codewords is done implicitly by the encoding/decoding algorithm. We introduce the notion of the \textit{base codebook}, which contains all codewords, i.e., the selected and non-selected ones.
\begin{definition} \label{def:base_cb}
A base codebook, denoted by $\tilde{\mathcal{C}}$, for the coding scheme from Sec. \ref{s:acoding} is a codebook which is used to compute the branching probabilities (\ref{eq:pb_s}), i.e.,
\begin{equation}
	P(a|\vec{s}) = \frac{N(\vec{s}a)}{N(\vec{s})} = \frac{| \{ \vec{c} \in \tilde{\mathcal{C}} \colon \vec{c}_1^{l(\vec{s})} = \vec{s} \wedge c_{l(\vec{s})+1}=a\}|}{| \{ \vec{c} \in \tilde{\mathcal{C}} \colon \vec{c}_1^{l(\vec{s})} = \vec{s}\}|},
\end{equation}
for any $a \in \mathcal{A}$ and any prefix $\vec{s}$.   
\end{definition}

The \textit{actual codebook} $\mathcal{C}$, is the codebook actually used by the encoder/decoder, i.e.,
\begin{equation}
\mathcal{C} = \{ \vec{c} \in \tilde{\mathcal{C}} \colon \exists {\vec{u} \in \{0,1\}^k} \text{  s.t. } \vec{c} = f_{\text{DM}}(\vec{u}) \},
\end{equation} 
where $f_{\text{DM}}$ is the encoder function. The actual codebook is a subset of the base codebook implicitly chosen by the encoding/decoding algorithm.

\addtolength{\topmargin}{-0.1in}
\addtolength{\textheight}{0.1in} 
\section{Multi-Composition Codebooks} \label{s:mcc}
Assume an arbitrary base codebook $\tilde{\mathcal{C}} \!=\! \{ \vec{c}_i \!\in\! \{0,1\}^n, i=1,\dotsc,N \}$. Using the coding scheme from Sec. \ref{s:acoding}, we are able to encode/decode data into/from codewords from $\tilde{\mathcal{C}}$. This involves computing the probabilities (\ref{eq:pb_s}). In a general case, finding (\ref{eq:pb_s}) entails evaluating $N(\vec{s})$ by counting the codewords from $\tilde{\mathcal{C}}$ which have the prefix $\vec{s}$. This is not feasible for large codebooks which we target. Therefore, we introduce a structure into the base codebook $\tilde{\mathcal{C}}$. 

We observe that $N(\vec{s})$ is easy to compute for base codebooks containing \textit{all codewords} of a single composition. Assume a binary composition $\vec{\gamma}_0 = [n-m, m]$ and the codebook $\tilde{\mathcal{C}} \!=\! \{ \vec{c} \!\in\! \{0,1\}^n \colon  \gamma(\vec{c}) = \vec{\gamma}_0\}$, i.e., the so-called $m$-out-of-$n$ codebook containing all codewords of Hamming weight $m$. For any prefix $\vec{s}$, we have
\begin{equation}
N(\vec{s}) = \binom{n-l(\vec{s})}{m-n_1(\vec{s})}.
\end{equation} 

Consequently, $N(\vec{s})$ is simple to compute for base codebooks containing \textit{all codewords} of multiple compositions. Assume a set of different binary compositions
\begin{equation} 
\Gamma = \{ \vec{\gamma}_1, \dots, \vec{\gamma}_I \}, \label{eq:composition_set}
\end{equation}
where $\vec{\gamma}_i = [n-m_i, m_i]$ for $i=1,\dotsc,I$, and the base codebook $\tilde{\mathcal{C}} \!=\! \{ \vec{c} \!\in\! \{0,1\}^n \colon  \vec{\gamma}(\vec{c}) \in \Gamma\}$. Here, we will refer to such a codebook as the MC base codebook. For any prefix $\vec{s}$, we have
\begin{equation}
N(\vec{s}) = \sum_{i=1}^{I} \binom{n-l(\vec{s})}{m_i-n_1(\vec{s})},
\end{equation}
which is easy to evaluate. Note that $N(\vec{s})$ depends only on two parameters $l(\vec{s})$, and $n_1(\vec{s})$, therefore it is also possible to store the precomputed values in a look-up-table (LUT). For implementation we need only to store the values of $P(0|\vec{s})$ (or $P(1|\vec{s})$). As $P(0|\vec{s})$ (or $P(1|\vec{s})$) depends only on $l(\vec{s})$ and $n_1(\vec{s})$, we need to store at most $n^2$ values.  For an arbitrary base codebook, $N(\vec{s})$ depends on the whole sequence $\vec{s}$ and it is not feasible to store all values in a LUT. Here, we will refer to a DM using the base MC codebook and the coding scheme from Sec. \ref{s:acoding} as the Multi-Composition Distribution Matcher (MCDM).


\begin{figure}
	\centering
	\begin{subfigure}[t]{0.48\columnwidth}%
	\centering
	\includegraphics[scale=0.7]{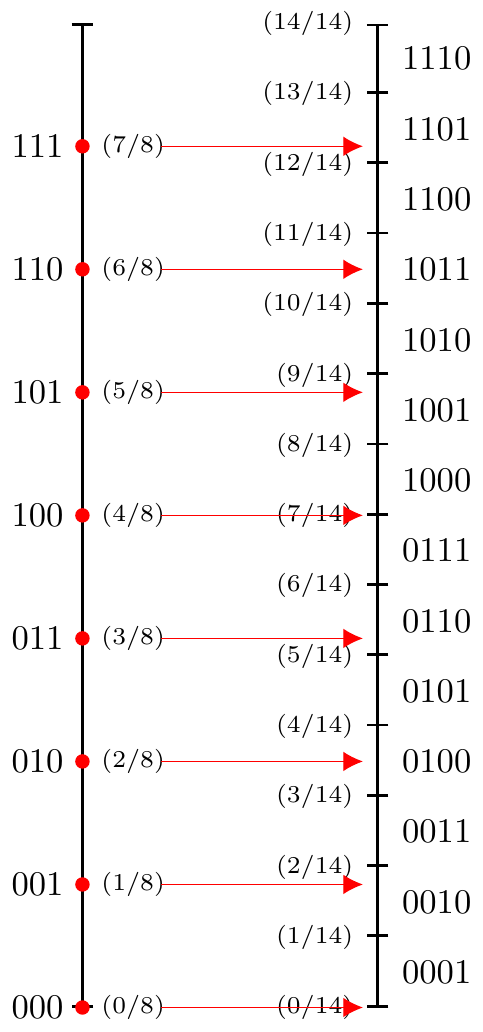}
		\caption{$[1,3]$-out-of-$4$ base codebook.}%
	\end{subfigure}%
	~ 
	\begin{subfigure}[t]{0.48\columnwidth}
	\centering	
	\includegraphics[scale=0.7]{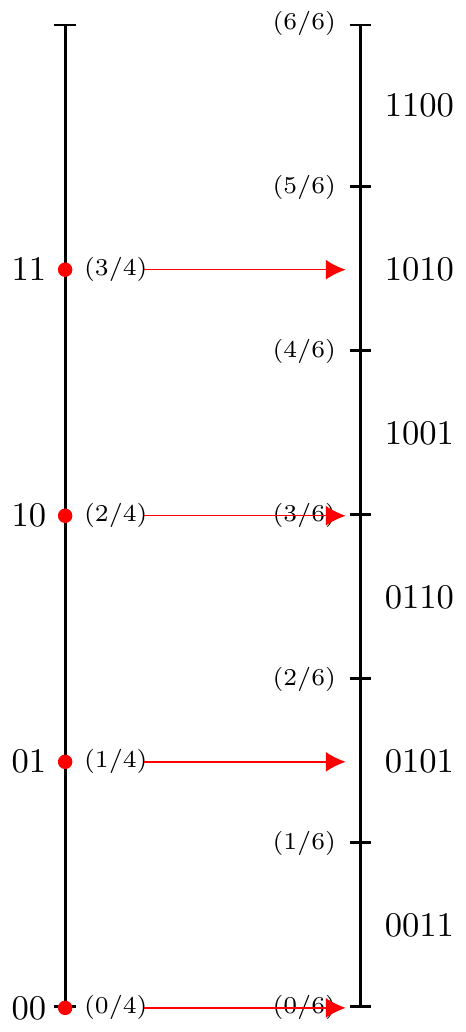}
		\caption{$2$-out-of-$4$ base codebook.}
		\label{fig:ccdm_intervals}	
	\end{subfigure}
	\caption{Encoding/decoding into the codewords from the base codebooks using the algorithm from Sec. \ref{s:acoding}.}
	\label{f:base_cb}
\end{figure}



\subsection{Some Special Cases} \label{ss:specialcases}
By selecting the set of compositions (\ref{eq:composition_set}) we can obtain a specific base codebook. 
\begin{definition} \label{def:opt_dm}
	A $[m_L,m_U]$-out-of-$n$ codebook is a codebook with codewords of Hamming weight at least $m_L$ and at most $m_U$, i.e., 
	\begin{equation}
	\tilde{\mathcal{C}} = \{ \vec{c} \in \{0,1\}^n \colon m_L \le n_1(\vec{c}) \le m_U \}.
	\end{equation}   
\end{definition}

Based on Definition \ref{def:opt_dm}, we have the following special cases.
\subsubsection{$m$-out-of-$n$ codebook} 
By selecting $\Gamma = \left\{[n-m,m]\right\}$ we obtain a CC codebook as in the CCDM. For CC codebook the probability in (\ref{eq:pb_s}) has a particularly simple form, i.e., 
\begin{equation}
P(1|\vec{s}) = \frac{m-n_1(\vec{s})}{n-l(\vec{s})}.
\end{equation}

\subsubsection{$[m\!-\!1,m]$-out-of-$n$ codebook}
By selecting $\Gamma = \left\{[n-m+1,m-1],[n-m,m]\right\}$ we obtain a codebook which contains two adjacent compositions. We refer to such a MCDM as 2C-MCDM. The probability  in (\ref{eq:pb_s}) also admits a simple form 
\begin{equation}
P(1|\vec{s}) = \frac{m-n_1(\vec{s})}{n-l(\vec{s})+1}.
\end{equation} 
As such, the CCDM can be changed into the 2C-MCDM by just changing the denominator in the applied model for arithmetic coding.

\subsubsection{$[0,m]$-out-of-$n$ codebook}
By selecting $\Gamma = \left\{ [n,0],\right.$ $ \left.[n-1,1], \dots, [n-m,m]\right\}$ we obtain a codebook which contains all sequences up to Hamming weight $m$. This is the optimal codebook for an ideal DM which can encode into a codebook of arbitrary size.

\begin{lemma} Assume the output alphabet $\mathcal{A}=\{0,1\}$, and the target probability such that $P_A(0)>P_A(1)$. Assume a DM with output $\mathbb{U}_{\mathcal{C}}$. The actual codebook $\mathcal{C}$ which minimizes the normalized KL divergence (\ref{eq:def_div}) 
\begin{enumerate}[label=(\alph*)]
\item consists of $2^k$ most likely codewords according to $P_A$, if  we require $|\mathcal{C}|=2^k$ for some $k$ \cite{b2bdm}. This DM can be implemented by the SMDM \cite{shelldm}. 
\item is a $[0,m]$-out-of-$n$ codebook for some $m$, if $|\mathcal{C}|$ is not constrained \cite{div_scaling}.
\end{enumerate}
\end{lemma}
\begin{proof}
See \cite[Sec. IV]{b2bdm}, \cite[Lemma 5]{div_scaling}.
\end{proof}

The MCDM with the $[0,m]$-out-of-$n$ codebook, which we will refer to as Opt-MCDM, can be seen as an approximation of the optimal DM from Lemma 1b. Note that when the $[0,m]$-out-of-$n$ codebook size is equal to $2^k$, the Opt-MCDM is the optimal DM (as in Lemma 1a). In practice, the Opt-MCDM offers close to optimal performance. The probability in (\ref{eq:pb_s}) is
$$P(1|\vec{s}) = \frac{\sum_{i=0}^{m} \binom{n-1-l(\vec{s})}{i-1-n_1(\vec{s})}}{\sum_{i=0}^{m} \binom{n-l(\vec{s})}{i-n_1(\vec{s})}},$$
which depends only on $l(\vec{s})$ and $n_1(\vec{s})$, and it can therefore be precomputed and stored in a LUT for efficient implementation.

\begin{figure}
	\centering
	\begin{subfigure}{\columnwidth}%
		\centering
	\includegraphics[]{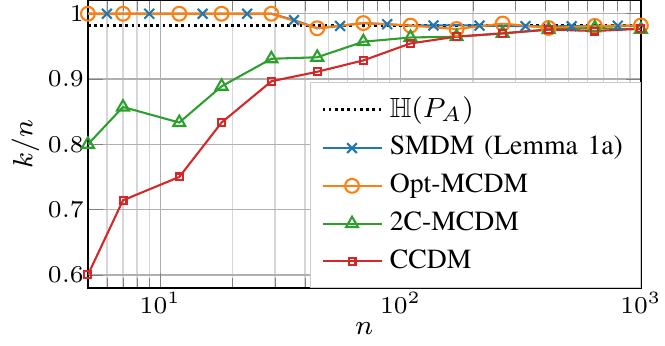}
		\vspace{-0.2cm}
		\caption{Matching rate $k/n$ vs $n$}%
		\label{f:results1a}%
	\end{subfigure}
	\vspace{0.35cm}
	
	\begin{subfigure}{\columnwidth}%
	\centering
	\includegraphics[]{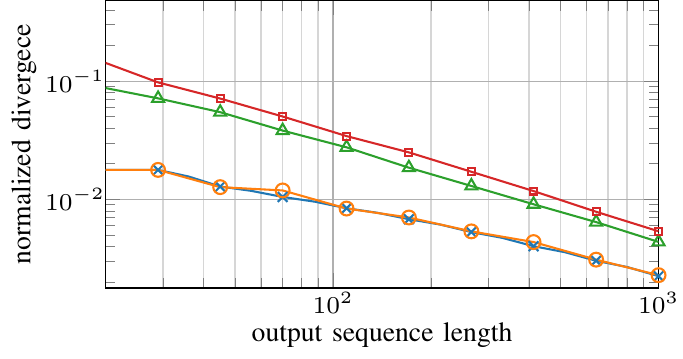}
		\vspace{-0.2cm}%
		\caption{Normalized KL divergence vs $n$}%
		\label{f:results1b}%
	\end{subfigure}
	\vspace{0.35cm}	
	
	\begin{subfigure}{\columnwidth}%
		\centering
		\includegraphics[]{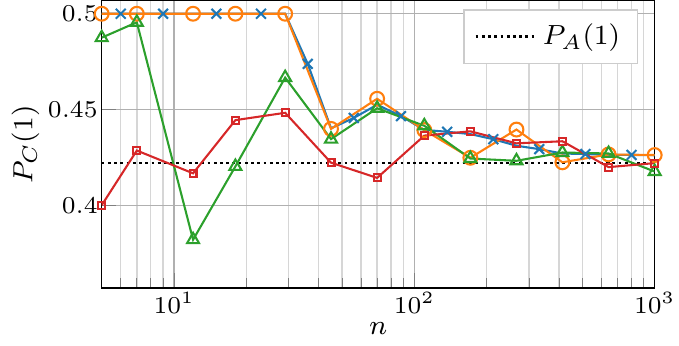}
		\vspace{-0.2cm}
		\caption{$P_{\mathcal{C}}(1)$ and the target $P_A(1)$ vs $n$}%
		\label{f:results1c}%
	\end{subfigure}%
	\vspace{0.1cm}	
	\caption{Parameters of the DMs with the optimal base codebooks for the target $P_A(1)=0.422$.}
	\label{f:results1}
	\vspace{-5mm}
\end{figure}

\subsection{Non-binary Case} \label{ss:nonbinary}
The coding scheme described in Sec. \ref{s:acoding} can be adapted for non-binary distribution matching. The probability model $P_{\vec{C}}$ on codewords' symbols can be obtained from the base codebook via equations analogous to (\ref{eq:pb_s}). An MC codebook contains all codewords from multiple compositions, and the expression for $N(\vec{s})$ becomes a sum of multinomials, which is admissible to evaluate or store for shorter codewords or fewer compositions in the base codebook. However, for the base codebooks with large number of compositions or long codewords, the storage/computation requirements can become prohibitive for large alphabet. Unless some structure is added when choosing the compositions, the MCDM is better suited for binary distribution matching, e.g., it can be used for non-binary distribution matching in combination with the bit-level distribution matcher \cite{7929328}.

\section{Results} \label{s:results}
\subsection{Distribution Matching Performance} \label{ss:dm_performance}
We compare the CCDM with the $m$-out-of-$n$ codebook, MCDM with the $[m-1,m]$-out-of-$n$ codebook (2C-MCDM), and MCDM with the $[0,m]$-out-of-$n$ codebook (Opt-MCDM). We use the binary output alphabet $\mathcal{A}$ and the target probability distribution $P_A$ with $P_A(1)=0.422$. We vary the output length $n$. For each of the DMs we find the base codebook $\tilde{\mathcal{C}}$ which minimizes the KL divergence $\frac{1}{n}\mathbb{D}(\mathbb{U}_{\tilde{\mathcal{C}}}||P_A^n)$. This is equivalent to finding
\begin{equation}
m^* = \argmin_m \sum_{\vec{c} \in \tilde{\mathcal{C}}(m)} \frac{1}{|\tilde{\mathcal{C}}(m)|} \log \frac{\frac{1}{|\tilde{\mathcal{C}}(m)|}}{P_A^n(\vec{c})}, 
\end{equation}
where $\tilde{\mathcal{C}}(m)$ is the $m$-out-of-$n$, $[m-1,m]$-out-of-$n$, and $[0,m]$-out-of-$n$ codebook for the CCDM, 2C-MCDM, and Opt-MCDM, respectively. Next, we apply the coding scheme as in Sec. \ref{s:acoding} to build DMs using the aforementioned optimized base codebooks. The results are presented in Fig. \ref{f:results1}.  The divergence and the empirical output distribution $P_{\mathcal{C}}$ were computed by enumerating all codewords for $n \le 30$ and for higher $n$ via Monte-Carlo sampling. The number of samples was chosen so that the relative error of estimates lies within $3\%$ with probability not smaller than $90\%$. The matching rate is evaluated exactly. We also compute the parameters of the optimal DM from Lemma 1a implemented by SMDM \cite{shelldm}. 

Fig. \ref{f:results1} shows the superior performance of the MC codebooks in terms of the matching rate and KL divergence. Opt-MCDM stays very close to the optimal DM from Lemma 1a. In Fig. \ref{f:results1a} the optimal DM achieves higher rate than the entropy of the target distribution $\mathbb{H}(P_A)$ for $n \le 30$. This is because for $n \le 30$, the optimal base codebook is the $[0,n]$-out-of-$n$ codebook which contains all codewords of length $n$. This demonstrates that the lowest KL divergence can be achieved by performing no distribution matching at all, and hence the optimal codebook has $P_{\mathcal{C}} \neq P_A$. For large $n$, $P_{\mathcal{C}}(1) \to P_A(1)$ for all DMs, as this is a necessary condition for vanishing divergence for $n \to \infty$, as observed in Sec. \ref{s:dm_fund}. In Fig. \ref{f:results1c} for $n=110$, $P_{\mathcal{C}}$ coincide for all DMs. However, in Fig. \ref{f:results1b} the Opt-MCDM achieves the lowest KL divergence thanks to the largest codebook, confirming the observations from Sec. \ref{s:dm_fund}.

\begin{figure}
	\includegraphics[]{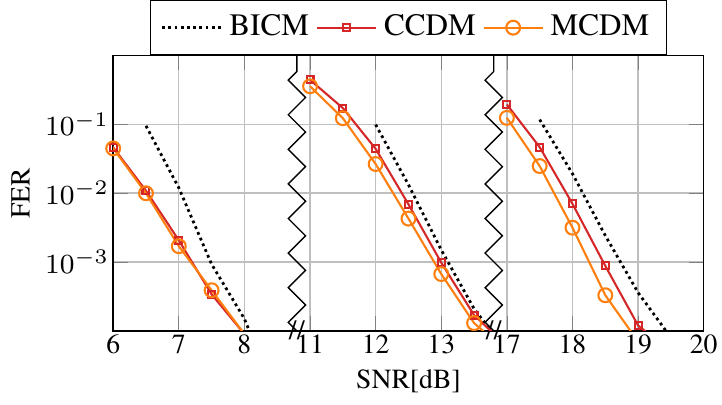}
	\vspace{-4mm}
	\caption{BICM, and PAS in the bit-level setup \cite{7929328} with CCDM, and MCDM, for rates $1.67$, $3.33$, $5.0$ b/CU. The channel is an additive white Gaussian noise channel.}
	\label{f:results2}
	\vspace{-6mm}
\end{figure}

\subsection{PAS Framework}
In practice, applying the Opt-MCDM instead of the CCDM for $n=110$ in a PAS communication system would mean ca. $3\%$ increase in the transmission rate (see Fig. \ref{f:results1a}).
Alternatively, assume we target the KL divergence $10^{-2}$. Instead of using one CCDM with $n=500$, we can use $5$ parallel Opt-MCDMs with $n=100$ to increase the throughput, as shown in Fig. \ref{f:results1b}.

Motivated by this considerations we apply the Opt-MCDM to a PAS system in a bit-level setup as in \cite{7929328}. We follow the steps exactly as in \cite{7929328} but instead of using the CCDM as a building block, we employ the Opt-MCDM for each bit-level. For each bit-level's target probability we find the optimal codebook as in Sec. \ref{ss:dm_performance}. We compare the results with \cite{7929328} employing the CCDMs, and the bit-interleaved coded-modulation (BICM) scheme without shaping \cite{669123}. For fair comparison, all schemes use the same WiMAX LDPC B-code of rate $5/6$ and codeword length of $576$ bits. LDPC decoder performs $50$ iterations. BICM operates with $3$ constellations: $4$-QAM, $16$-QAM, and $64$-QAM, which corresponds to transmission rates of  $1.67$, $3.33$, $5.0$ bits per channel use (b/CU), respectively. Shaping schemes use the $256$-QAM constellation and match the BICM transmission rates by applying appropriate transmit signal distributions. Frame error rate (FER) versus signal-to-noise ratio (SNR) curves are presented in Fig. \ref{f:results2}. By employing the MCDM instead of the CCDM, we gain  $0.02$, $0.11$, and $0.23$ dB, at FER=$10^{-3}$ for rates $1.67$, $3.33$, $5.0$ b/CU, respectively.

%

\vspace{-1mm}
\section{Conclusions} \label{s:conclusions}
In this work, we presented arithmetic coding based distribution matcher which uses multi-composition codes. The multi-composition distribution matcher generalizes the state of the art constant-composition distribution matcher, and is able to achieve higher matching rate and lower KL divergence.

\vspace{-1mm}
\section{Acknowledgments}
Part of this work has been performed in the framework of the Horizon 2020 project ONE5G (ICT-760809) receiving funds from the European Union. The authors would like to acknowledge the contributions of their colleagues in the project, although the views expressed in this contribution are those of the authors and do not necessarily represent the project.

The authors would like to thank Onurcan Iscan, Ronald B{\"{o}}hnke, Najeeb Ul Hassan from Huawei Technologies, for discussions and helpful comments for improving the manuscript.

\bibliographystyle{IEEEtran}
\vspace{-1mm}
\small{
\bibliography{references}
}

\end{document}